\documentclass[conference, 10pt]{IEEEtran}
\usepackage{setspace}
\usepackage{amsmath}
\usepackage{amssymb}
\usepackage{mathrsfs}
\usepackage{amsthm}
\usepackage{multirow}
\usepackage{color}
\usepackage{array}
\usepackage{float}
\usepackage{url}
\usepackage{comment}
\usepackage{enumerate}
\usepackage{graphicx}
\usepackage{enumitem}
\usepackage{tabularx}
\usepackage{makecell}
\usepackage[capitalize]{cleveref}
\usepackage{fixfoot}
\usepackage{caption}
\usepackage[style=ieee]{biblatex}
\usepackage{algorithm}
\usepackage[noend]{algorithmic}
\usepackage[table]{xcolor}
\usepackage{blkarray}
\usepackage{balance}

\usepackage{stmaryrd}

\usepackage[top=0.55in, bottom=0.871in, left=0.6in, right=0.621in]{geometry}

\IEEEoverridecommandlockouts

\theoremstyle{plain}
\newtheorem{theorem}{Theorem}
\newtheorem{corollary}[theorem]{Corollary}
\newtheorem{lemma}[theorem]{Lemma}
\newtheorem{proposition}[theorem]{Proposition}

\theoremstyle{definition}
\newtheorem{definition}{Definition}
\newtheorem{example}[definition]{Example}

\newtheorem{remark}[definition]{Remark}

\newcommand{\RM}{{\text{\small $\mathbb{RM}$}}}

\newcommand{\numberset}{\mathbb}

\newcommand{\F}{\numberset{F}}
\newcommand{\supp}{\textup{supp}}

\newcommand{\mC}{\mathcal{C}}
\newcommand{\mU}{\mathcal{U}}
\newcommand{\mW}{\mathcal{W}}
\newcommand{\mP}{\mathcal{P}}
\newcommand{\mR}{\mathcal{R}}
\newcommand{\mY}{\mathcal{Y}}
\newcommand{\mD}{\mathcal{D}}

\newcommand{\ps}{\oplus_\textnormal{P}}

\DeclareMathOperator{\wtH}{wt_H}

\usepackage[style=ieee]{biblatex}
\DeclareMathAlphabet{\mathbfsl}{OT1}{ppl}{b}{it} 

\newcommand{\anina}[1]{{\color{magenta}(Anina: #1)}}

\usepackage{balance}

\newcolumntype{Y}{>{\centering\arraybackslash}X} 

\title{\vspace{-1ex} {\huge{{\textbf{Convertible Codes for Data and Device Heterogeneity}}}
}\vspace{-1ex}}
\author{%
   \IEEEauthorblockN{\textbf{Anina Gruica}\IEEEauthorrefmark{1}, \textbf{Benjamin Jany}\IEEEauthorrefmark{2}, and \textbf{Stanislav Kruglik}\IEEEauthorrefmark{1}
   \thanks{The work of Anina Gruica was supported by the Villum Fonden under Grant VIL52303 and by the EuroTech Visiting Researcher Program 2025. The work of Benjamin Jany was supported by the Casimir Institute at the Eindhoven University of Technology and by the EuroTech Visiting Researcher Program 2025.}
   }
    \IEEEauthorblockA{\IEEEauthorrefmark{1}%
    Technical University of Denmark, Kongens Lyngby, Denmark
    }
    \IEEEauthorblockA{\IEEEauthorrefmark{2}%
    Eindhoven University of Technology, the Netherlands}
    \texttt{\{anigr, stakr\}@dtu.dk, b.jany@tue.nl \vspace{-1ex}}
 \vspace{-2ex}}
\IEEEaftertitletext{}

\begin{document}
\date{}


\maketitle


\begin{abstract}
Distributed storage systems must handle both data heterogeneity, arising from non-uniform access demands, and device heterogeneity, caused by time-varying node reliability. In this paper, we study convertible codes, which enable the transformation of one code into another with minimum cost in the merge regime, addressing the latter. We derive general lower bounds on the read and write costs of linear code conversion, applicable to arbitrary linear codes. We then focus on Reed–Muller codes, which efficiently handle data heterogeneity, addressing the former issue, and construct explicit conversion procedures that, for the first time, combine both forms of heterogeneity for distributed data storage.
\end{abstract}

\setstretch{0.97}

\section{Introduction}


Erasure codes are a core component of most existing large-scale distributed storage systems and provide low-overhead tolerance to node failures~\cite{balaji2018erasure, shen2025survey}. In this approach, data are encoded using a linear code with codeword length $n$ and dimension $k$, and the resulting code symbols are distributed across different nodes. The encoding parameters are typically selected based on the current failure rate of the system. However, the failure rates of storage devices may vary over time. Adapting the code rate in response to such variations can lead to significant savings in storage space and, consequently, operational costs~\cite{kadekodi2019cluster}. This property is referred to as {\em device heterogeneity}. Nevertheless, modifying code parameters in a storage system typically requires re-encoding the stored data and, in the case of non-systematic codes, decoding all message symbols. Such procedures involve accessing and reading a large number of stored symbols, which can be costly and may outweigh the benefits of conversion~\cite{maturana2022convertible}.

To overcome the abovementioned issues, Maturana and Rashmi recently proposed a theoretical code conversion framework that defines the process of re-encoding data from one code to another~\cite{maturana2022convertible}. They also introduced convertible codes that minimize conversion costs while maintaining desired properties, such as the MDS property or locality. Code conversion has been studied in two regimes: the merge regime, in which codewords of the initial code are merged into a smaller number of codewords of the final code, and the opposite split regime. The main focus of this line of work is to optimize the conversion cost either with respect to the number of symbols being read and written during conversion (the so-called access cost~\cite{maturana2022convertible}) or with respect to the amount of information downloaded (the so-called bandwidth cost~\cite{maturana2023bandwidth}). Another line of recent research focuses on security against passive eavesdroppers with access to a subset of stored codeword symbols~\cite{zhang2025secure}. However, the works mentioned above, as well as subsequent papers, primarily focus on two specific families of codes employed in distributed storage systems—codes with locality and MDS codes, along with their generalizations; see~\cite{ge2025locally, ge2024mds, singhvi2025tight, kong2024locally, ramkumar2025mds, shi2025bounds,chopra2024low, krishnan2025linear, wang2025lower} and references therein.

However, in distributed storage systems, different data objects may attract different levels of user interest and, consequently, may require different target access rates. This property is referred to as {\em data heterogeneity}. At the same time, the total request rate that can be served by each storage node is limited. The problem of characterizing the set of achievable access rates for multiple files under such constraints was recently introduced under the notion of the service rate region~\cite{emina21}. This framework introduces a new perspective on distributed storage system design by characterizing their efficiency and scalability. Recent results provide a detailed analysis of the service rate region of Reed–Muller codes, highlighting their suitability for distributed storage systems under so-called heterogeneous requests~\cite{ly2025service}. At the same time, Reed–Muller codes have largely remained outside the focus of research on the code conversion problem, making them suitable primarily for deployment in systems with unchanged device characteristics.

In this paper, we fill this gap by investigating the conversion properties of Reed–Muller codes in the merge regime, offering the first construction of convertible codes that support both data and device heterogeneity. The key elements of our construction are the structural properties of Reed–Muller codes and the Plotkin construction. To evaluate the performance of our conversion procedure, we, for the first time, consider the conversion of general linear codes and derive bounds on the write and read costs by analyzing the matrix representation of the conversion process and by extending puncturing bounds for the conversion of codes with locality from~\cite{ge2025locally, shi2025bounds} to the general case. These bounds provide a first-order estimate of the performance of code conversion for general families of codes used in distributed storage systems, as obtaining tighter bounds necessitates exploiting the structural properties of the codes under consideration. 

The rest of the paper is organized as follows. Section~\ref{sec2} provides the required preliminaries. Section~\ref{sec:gen} presents general bounds on the access cost in the merge regime. Convertible Reed–Muller codes are introduced in Section~\ref{sec:RM} and compared with the derived bounds in Section~\ref{sec:comp}. The paper concludes in Section~\ref{sec:concl}.

\section{Convertible Codes in the Merge Regime}\label{sec2}

\subsection{Preliminaries}

In this paper, we let $q$ be a prime power and we denote by $\F_q$ the finite field of $q$ elements. Moreover, unless specified otherwise, we let $k$ and $n$ be integers with $1\le k \le n$. We denote the set $\{1,\dots, n\}$ by $[n]$.

\begin{definition}
    A \textbf{linear code} $\mC$ of \textbf{dimension} $k$ is an $\F_q$-linear subspace of $\F_q^n$ of $\F_q$-dimension~$k$. The \textbf{minimum distance} of~$\mC$ is
    \begin{align*}
        d(\mC) := \min\{\wtH(x) : x \in \mC, x \ne 0\},
    \end{align*}
    where $\wtH(x):= |\supp(x)|$ is the \textbf{Hamming weight} of $x=(x_1,\dots, x_n) \in \F_q^n$ and $\supp(x) := \{i \in [n] :  x_i \ne 0\}$ its (Hamming) support. We say that a code $\mC$ in $\F_q^n$ of dimension~$k$ is an $[n,k]_q$-code. If additionally $\mC$ has minimum distance $d$, then we say it is an $[n,k,d]_q$-code.
\end{definition}


In this paper, we study the \textit{code conversion problem} in the \textit{merge regime} following the theoretical framework introduced in~\cite{maturana2022convertible}, and the generalization thereof from~\cite{ge2024mds}. From now on, we let $\lambda \ge 1$ be an integer and without loss of generality consider merge of several codes into one. Moreover, we let $1 \le k_F \le n_F$, $k_\textbf{I}:=(k_{I_1},\dots,k_{I_\lambda})$ and $n_\textbf{I}:=(n_{I_1},\dots,n_{I_\lambda})$,  $1 \le k_{I_i} \le n_{I_i}$ be integers, for $i \in [\lambda]$, and $\sum_{i=1}^{\lambda} k_{I_i}=k_F$. We define a \textit{convertible code} formally as follows.

\begin{definition}
A \textbf{convertible code} in the merge regime with parameters $(n_\textbf{I},k_\textbf{I},n_F,k_F)$ consists of
    \begin{itemize}
        \item[(i)] $\lambda$ \textbf{initial codes} $\mC_{I_1},\dots,\mC_{I_\lambda}$, where $\mC_{I_i}$ is an $[n_{I_i}, k_{I_i}]_q$ code for all $i \in [\lambda]$;
        \item[(ii)] a \textbf{final code} $\mC_F$ with parameters $[n_F,k_F]_q$;
        \item[(iii)] a \textbf{conversion procedure (map)} $\sigma: \prod_{i=1}^{\lambda} \mC_{I_i} \longrightarrow \mC_F$  
    \end{itemize}
If we let $\mC_\textbf{I} :=\prod_{i=1}^{\lambda}\mC_{I_i}$, then we denote the convertible code that merges the initial codes $\mC_{I_1},\dots,\mC_{I_\lambda}$ into the final code $\mC_F$ using the conversion procedure $\sigma$, by $(\mC_\textbf{I},\mC_F,\sigma)$.
\end{definition}
The main objective in the theory of convertible codes is to merge several \emph{small} codes used in a distributed storage system into a larger code without fully decoding them. There are two main performance metrics, namely bandwidth and access cost, and in this paper we focus on the latter. 

To better understand the access cost of a conversion procedure, we classify its symbols into three categories. 

\begin{definition} \label{def:clasym}
Let $(\mC_\textbf{I},\mC_F,\sigma)$ be a convertible code with parameters $(n_\textbf{I},k_\textbf{I},n_F,k_F)$. Then 
\begin{itemize}
    \item[(i)] for $i \in [\lambda]$, the \textbf{unchanged symbols} from $\mC_{I_i}$, denoted by $\mU_{i}$, are symbols from $\mC_{I_i}$ that remain unchanged under $\sigma$ (possibly in different locations), and the unchanged symbols from $\mC_{\textbf{I}}$ are $\mU := \bigcup_{i \in [\lambda]} \mU_i$;
    \item[(ii)]  \textbf{new symbols} in $\mC_F$, denoted by $\mW$, are symbols in $\mC_F$ that are not inherited from an unchanged symbol in $\mC_\textbf{I}$;
    \item[(iii)] for $i \in \lambda$ the \textbf{read symbols} from $\mC_{I_i}$, denoted by $\mR_{i}$, are symbols from $\mC_{I_i}$ that are used to determine a new symbol, and the read symbols from $\mC_{\textbf{I}}$ are $\mR := \bigcup_{i \in \lambda}\mR$.

\end{itemize}
Clearly, we have the following relation between the number of new and unchanged symbols $|\mW| = n_F - |\mU|$.
\end{definition}

We can now formally define it the access cost.
\begin{definition}
Given a convertible code $(\mC_\textbf{I},\mC_F,\sigma)$. The \textbf{read cost} of the conversion procedure counts the number of symbols that are read during the procedure, i.e., $|\mR|$.
The \textbf{write access} of the conversion procedure counts the number of new symbols in the final code that must be written, i.e., $|\mW|$.
The \textbf{access cost} of the conversion procedure is the sum of the read and write costs.
\end{definition}

\begin{remark} \label{rem:default}
Let $(\mC_\textbf{I},\mC_F,\sigma)$ be a convertible code with parameters $(n_\textbf{I},k_\textbf{I},n_F,k_F)$. The default conversion procedure reads $k_{I_i}$ code symbols from each inital code $\mC_{I_i}$, respectively, decodes the information, leaves these read symbols unchanged, and writes the new redundancy symbols (there are $n_F-k_F$ of them) to achieve the desired parameters of $\mC_F$. This results in an access cost of $n_F-k_F+\sum_{i=1}^{\lambda}k_{I_i}=n_F-k_F+k_F=n_F$.
\end{remark}

In theory, there are no assumptions on the conversion procedure from $\mathcal{C}_{\textbf{I}}$ to $\mathcal{C}_F$. However, to facilitate the understanding of convertible codes and to derive bounds for the read and write costs, we assume in what follows that $\sigma$ is a linear map. Furthermore, we assume that $\sum_{i=1}^{\lambda} k_{I_i} = k_F$, hence $\sigma$ is an isomorphism.  Considering non-linear conversions is an interesting avenue for future investigation.
\begin{example} \label{ex:conv}
   Let $\mC_{\mathbf{I}} := \mC_{I_1} \times \mC_{I_2} \leq \F_2^6$ and $\mC_F \leq \F_2^5$, where $\mC_{\mathbf{I}}$ and $\mC_F$ have respective generator matrices $G_{\mathbf{I}}$ and $G_F$ as defined below:
   \[
   G_{\mathbf{I}}
   :=
   \begin{bmatrix}
      1 & 0 & 1 & 0 & 0 & 0 \\
      0 & 1 & 1 & 0 & 0 & 0 \\
      0 & 0 & 0 & 1 & 1 & 0 \\
      0 & 0 & 0 & 0 & 1 & 1
   \end{bmatrix},
   \;
   G_F :=
   \begin{bmatrix}
      1 & 0 & 0 & 0 & 1 \\
      0 & 1 & 0 & 0 & 1 \\
      0 & 0 & 1 & 0 & 1 \\
      0 & 0 & 0 & 1 & 1
   \end{bmatrix}.
   \]
   An efficient conversion procedure $\sigma : \F_2^6 \rightarrow \F_2^5$ is as follows:
   \[
\sigma((x_1,x_2,x_3,x_4,x_5,x_6)) = (x_1, x_2, x_4,x_5,x_3+x_6).
   \]
   Note that $\sigma$ leaves $4$ symbols unchanged, $x_1, x_2, x_4, x_5$, and must write only the last symbol of the final code. Hence the writing cost is $1$. In order to write the last symbol of the final code, the conversion procedure must read $x_3$ and $x_6$, hence the reading cost is $2$. This gives a total access cost of $3$.
 By Remark~\ref{rem:default}, the default approach has an access cost of 5.
\end{example}

We note that a linear conversion procedure can be represented by a conversion matrix, whose properties are formalized below.
\begin{lemma} \label{lem:matrix}
Let $(\mC_\textbf{I},\mC_F,\sigma)$ be a convertible code with parameters $(n_\textbf{I},k_\textbf{I},n_F,k_F)$, and for $i \in [\lambda]$ let $G_{I_i} \in \F_q^{k_{I_i} \times n_{I_i}}$ be a generator matrix of $\mC_{I_i}$, respectively, and let $G_{F} \in \F_q^{k_F \times n_F}$ be a generator matrix of $\mC_F$. Consider the block matrix
\begin{align*}
    G_{\textbf{I}} = \begin{bmatrix} G_{I_1} & 0 & \cdots & 0 \\ 0 & G_{I_2} & \cdots  & 0 \\
    \vdots & \vdots && \vdots \\ 0 & 0 & \cdots &G_{I_{\lambda}}\end{bmatrix} \in \F_q^{k_F \times \sum_{i=1}^{\lambda} n_{I_i}}.
\end{align*}
If the conversion procedure $\sigma$ is linear, then it can represented as a matrix $\mY$ with the properties:

\begin{enumerate}
    \item[1)] $    G_{\textbf{I}} \cdot \mY = G_F, \quad \mY \in \F_q^{\sum_{i=1}^{\lambda} n_{I_i} \times n_F}.$ 
    \item[2)] the unchanged symbols of $\sigma$ are in one-to-one correspondence with the supports of columns of $\mY$ of weight 1.
    \item[3)] the read symbols are in one-to-one correspondence with the supports of columns of $\mY$ of weight at least 2.
\end{enumerate}


We call such a matrix $\mathcal{Y}$ a \textbf{conversion matrix} of $(\mC_I, \mC_F, \sigma)$.
\end{lemma}

\section{General Bounds on Write and Read Costs} \label{sec:gen}

In this section, we fix the initial codes $\mC_\textbf{I}=\prod_{i=1}^{\lambda}\mC_{I_i}$, where $\mC_{I_i}$ is a $[n_{I_i},k_{I_i}]_q$-code, and the final $[n_F,k_F]_q$-code $\mC_F$. We further assume for $i \in [\lambda]$ that $\mC_{I_i}$ has minimum distance $d_{I_i}$, the dual of $\mC_{I_i}$, denoted by $\mC_{I_i}^\perp$, has minimum distance $d_{I_i}^\perp$, $\mC_F$ has minimum distance $d_F$, and the dual of $\mC_F$, denoted by $\mC_F^{\perp}$, has minimum distance $d_F^\perp$. 

\subsection{Bounds on the Write Cost} \label{subsec:unchanged}
We give upper bounds on the number of unchanged symbols in a convertible code, implying lower bounds on the write cost.

\begin{proposition} \label{prop:unchangedupper1}
    Let $(\mC_{\textbf{I}}, \mC_F, \sigma)$ be a convertible code. For $i \in [\lambda]$, let $\mU_i$ be the unchanged symbols of $\mC_{I_i}$. We have
    \begin{align*}
        |\mU_i| \le \min\{n_{I_i},n_F-d_F-\sum_{j \ne i} k_{I_j}+1\}.
    \end{align*}
\end{proposition}
\begin{proof}
Without loss of generality, assume $i=1$. Let $\Tilde{\mC} := \sigma(0,\mC_{I_2},\dots,\mC_{I_{\lambda}}) \leq \mC_I $. Clearly, $d(\Tilde{\mC}) \ge d_F$ and $\dim(\Tilde{\mC})=\sum_{j \ne 1} k_{I_j}$. Let $N$ be the set of symbols in the final codeword that are either unchanged and come from $\mC_{I_1}$, or depend only on the read symbols from $\mC_{I_1}$, and let $N^c=[n_F]\setminus N$. Let $v \in \Tilde{\mC}$. Clearly, $v_{N}=0$, and so $\Tilde{\mC}_{|N}=0$. This implies that $d(\Tilde{\mC}) = d(\Tilde{\mC}_{|N^c})$ and also $\dim(\Tilde{\mC})=\dim(\Tilde{\mC}_{|N^c})=\sum_{j \ne 1} k_{I_j}$. By applying the Singleton bound to $\Tilde{\mC}_{|N^c}$ we have:
\begin{align*}
    d_F \le d(\Tilde{\mC}) \le (n_F-|N|)-\sum_{j \ne 1} k_{I_j}+1.
\end{align*}
As a result, $|N| \le n_F-d_F-\sum_{j \ne 1} k_{I_j}+1$. Clearly, $|\mU_{1}| \le |N|$ and thus, $|\mU_1|\le n_F-d-\sum_{j \ne 1} k_{I_j}+1$. We also trivially have $|\mU_{1}| \le n_{I_1}$, proving the desired result.
\end{proof}

A sharper (but more restrictive) bound that also takes into account the dual distance of the final code is given below.

\begin{proposition} \label{prop:unchangedupper2}
    Let $(\mC_{\textbf{I}}, \mC_F, \sigma)$ be a convertible code. For $i \in [\lambda]$, let $\mU_i$ be the unchanged symbols of $\mC_{I_i}$. If $d^{\perp}_F > k_{I_i} +1$ then $|\mU_i| \leq  k_{I_i}.$ In particular, if $d^{\perp}_F > k_{I_i} +1$ for all $i \in [\lambda]$ then 
    $$|\mU|  \leq\sum_{i=1}^{\lambda} k_{I_i} =  k_F,$$
    where $\mU = \bigcup_{i=1}^{\lambda} \mU_{i}$ is the total number of unchanged symbols.
\end{proposition}
\begin{proof}
    First note that any $t \leq d^{\perp}_F-1$ columns of a generator matrix of $\mC_F$ must be linearly independent. 
    Let $i \in [\lambda]$. If $|\mU_{I_i}| \geq k_{I_i}+1$, then because $\dim(\mC_{I_i}) = k_{I_i}$ it must be that the columns indexed by elements of $|\mU_i|$ in $\mC_F$ are linearly dependent. However, that would imply that $d_F \leq k_{I_i} +1$, a contradiction. The statement on the total number of unchanged symbols immediately follows from summing over the unchanged symbols from each one of the initial codes.
\end{proof}
\begin{remark}
It can be shown that, in the case of MDS codes, the derived bound coincides with that of~\cite[Lemma 11]{maturana2022convertible}.
\end{remark}
In Proposition~\ref{prop:unchangedupper2}, it is not possible to omit the requirement that $d^{\perp}_F > k_{I_i} + 1$ to derive a non-trivial bound on the number of unchanged coordinates without involving additional parameters of the code. This is illustrated in the following example. 
\begin{example}   
 Let $\lambda=2$ and let {$\mathcal{C}_{I_1} = \mathcal{C}_{I_2}
= \left\langle \begin{bmatrix} 1 & 1 \end{bmatrix} \right\rangle, \mathcal{C}_{\mathbf{I}} = \mathcal{C}_{I_1} \times \mathcal{C}_{I_2}$, and $\mathcal{C}_F = \left\langle 
\begin{bmatrix} 1 & 0 & 1 \end{bmatrix},
\begin{bmatrix} 0 & 1 & 0 \end{bmatrix}
\right\rangle.$} Note that $d_F^{\perp} = 2 = k_{I_i} +1$ for $i \in [2]$. The code $\mC_F$ can be obtained by leaving the two coordinates of the generator matrix of $\mathcal{C}_{I_1}$ and the first coordinate of the generator matrix of $\mathcal{C}_{I_2}$ unchanged.
\end{example}
We recall the definitions of shortening and puncturing of a code, which we will need later to derive a lower bound on the number of unchanged symbols.
\begin{definition} \label{short/punct}
    Let $\mC \subseteq \F_q^n$ be a code and $S \subseteq [n]$. 
    \begin{itemize}
        \item[(ii)] The \textbf{puncturing} of $\mC$ onto $S$ is $\mC_{|S}:=\{ \pi_S(x): x \in \mC\}$ where $\pi_S:\F_q^n \to \F_q^{|S|}$ is the projection onto the coordinates indexed by $S$.
        \item[(iii)] Let $\mC(S):=\{x \in \mC : \supp(x) \subseteq S\}$. The \textbf{shortening} of $\mC$ by the set $S$ is the code $\pi_S(\mC(S))$.
    \end{itemize}
\end{definition}

We will use the following well-known result, relating the dual of the shortening of a code with the punctured dual code; see for example~\cite{huffman2010fundamentals}.
\begin{lemma} \label{lem:dual}
    Let $\mC \subseteq \F_q^n$ be a code and let $\mC^\perp \subseteq \F_q^n$ be its dual. For a set $S \subseteq [n]$ we have $(\mC(S))^\perp = \mC^\perp_{|S}.$
\end{lemma}

\begin{proposition} \label{prop:unchangedupperbound}
Let $(\mC_{\textbf{I}}, \mC_F, \sigma)$ be a convertible code. 
If $\lambda \geq 2$, then for $i \in [\lambda]$, we have
    $$d_F \leq  |\mU_i \cup \mW| - k_{I_i} + 1 \leq n_F - k_F +1.$$
    In particular,
    \begin{equation}\label{eqt:unchboundi}
     \sum_{j \neq i}| \mU_j | \geq \sum_{j \neq i}k_{I_j}.
    \end{equation}
\end{proposition}
\begin{proof}
    Without loss of generality let $i=1$. Consider the subcode $\Tilde{\mC} := \sigma(\mC_1, 0, \ldots, 0) \leq \mC_F$. Note that $\supp(\Tilde{\mC}) \subseteq \mU_1 \cup \mW$ and $\dim(\Tilde{\mC}) = k_{I_1}$. Let $S=\mU_i\cup \mW$. Clearly, $d(\Tilde{\mC}) \ge d_F$, and so applying the Singleton bound to $\Tilde{\mC}$ we get 
    $$d_F \leq d(\Tilde{\mC}) \le  |S| - \dim(\mC(S)) + 1 =  |S| - k_{I_1} + 1.$$
    Now using shortening and puncturing duality from Lemma~\ref{lem:dual} we get that
    $$|S| - \dim(\mC(S)) = \dim (\mC^{\perp}_{F|S}) \leq \dim \mC_F^{\perp} = n_F - k_F.$$
    Hence, this gives us 
    $$d_F \leq |S| - k_{I_1} + 1 \leq n_F - k_F +1.$$
    Next, note that $n_F = \sum_{i=1}^{\lambda}|\mU_i| + |\mW|$.
    This implies for $\lambda \geq 2$ that 
    $$\sum_{j \neq 1} k_{I_j} \leq \sum_{j \neq 1} |\mU_j|$$
    proving the desired result.
\end{proof}

As a result we obtain the following lower bound on the total number of unchanged symbols. Note that this is the first, up to our knowledge, non-trivial lower bound on the number of unchanged symbols.
\begin{lemma}\label{cor:lowbndunch}
    Let $(\mC_{\textbf{I}}, \mC_F, \sigma)$ be a convertible code with $\lambda \geq 2$. Then $k_F \leq |\mU|$.
\end{lemma}

\begin{proof}
    From \eqref{eqt:unchboundi}, we know that for all $i \in [\lambda]$ we have
    $$ \sum_{j \neq i}k_{I_j} \leq  \sum_{j \neq i}| \mU_j |.$$
    Thus summing over all $i \in [\lambda]$ we have that 
    $$(\lambda-1)k_F = \sum_{i \in [\lambda]} \sum_{j \neq i}k_{I_j}   \leq \sum_{i \in [\lambda]} \sum_{j \neq i} |\mU_j| = (\lambda-1)|\mU|.$$
    Hence the statement follows. 
\end{proof}

\cref{cor:lowbndunch} has strong implications, as it precisely determines the number of unchanged symbols for certain parameters.

\begin{corollary}
    Let $(\mC_{\textbf{I}}, \mC_F, \sigma)$ be a convertible code with $d_F^{\perp} > k_{I_i} +1$ for all $i \in [\lambda]$. Then $|\mU_i| = k_{I_i}$ for all $i \in [\lambda]$ and $|\mU| = k_F$.
\end{corollary}
\begin{proof}
    Follows from \cref{prop:unchangedupper2} and \cref{prop:unchangedupperbound}.
\end{proof}
\subsection{Bounds on the Read Cost} \label{subsec:read}
In this subsection, we provide general bounds on the read cost of convertible codes based on recently proposed bounds for codes with locality from~\cite{ge2025locally, shi2025bounds} and the following classic result on code puncturing; see, for example,~\cite{huffman2010fundamentals}. First, we relate the number of read symbols to the dimension of the initial code, the distance of the final code, and the number of unchanged symbols. Then, by introducing bounds on the latter, we obtain a lower bound on the number of read symbols.

\begin{lemma} \label{lem:punc}
    Let $\mC \le \F_q^n$ be a code of dimension $k$ and minimum distance $d$. For a set $S \subseteq [n]$ with $|S| \ge n-d+1$, the punctured code $\mC_{\mid S}$ has dimension~$k$.
\end{lemma}

\begin{proposition}\label{prop:readlower}
    Let $(\mC_{\textbf{I}}, \mC_F, \sigma)$ be a convertible code. Let $i \in [\lambda]$, let $\mR_i \subseteq [n_F]$ be the read symbols from $\mC_{I_i}$ and define $\delta_i:=|\mU_i|-d_F+1$. We have
    \begin{align*}
        |\mR_i| \ge \begin{cases}
            k_{I_i} \quad &\textnormal{if $\delta_i \le 0$,} \\
            k_{I_i}-\delta_i \quad &\textnormal{if $\delta_i > 0$.}
        \end{cases}
    \end{align*}
Moreover, if for $i \in [\lambda]$ we have $d_F > n_I-k_{I_i}+1$ then $\delta_i \le 0$.
\end{proposition}
\begin{proof}
We prove the statement for $i=1$, but the other cases can be done analogously. Consider the subcode $\Tilde{\mC} := \sigma(\mC_{I_1}, 0, \ldots, 0) \leq \mC_F$. Clearly, $\Tilde{\mC}$ has dimension $k_{I_1}$ and minimum distance $d(\Tilde{\mC}) \ge d_F$. We split the proof into two cases.

\textit{Case $\delta_i \le 0$:} Suppose that $|\mU_1|-d_F+1 \le 0$. Then 
$$n_F-|\mU_1| \ge n_F-d_F+1 \ge n_F-d(\Tilde{\mC})+1.$$
Applying Lemma~\ref{lem:punc} to the set $S:=[n_F] \setminus \mU_1$ gives $\dim(\Tilde{\mC}_{\mid S}) = k_{I_1} = \dim(\Tilde{\mC})$. Since $S \cap \mU_1 = \emptyset$, the non-zero coordinates in $\Tilde{\mC}_{\mid S}$ are functions of the coordinates in $\mR_1$ and so we obtain $|\mR_1| \ge \dim(\Tilde{\mC}_{\mid S}) =k_{I_1}$. 

\textit{Case $\delta_i > 0$:} Suppose that $|\mU_1|-d_F+1 > 0$. Arbitrarily choose any $d_F-1$ elements from $\mU_1$ and remove them from $[n_F]$ resulting in a set $S$ of size $n_F-d_F+1$. We have $|S|=n_F-d_F+1 \ge n_F-d(\Tilde{\mC})+1$. From Lemma~\ref{lem:punc} we obtain $\dim(\Tilde{\mC}_{\mid S}) = k_{I_1}$. Let $S_1=S \cap \mU_1$. Clearly, $|S_1| = |\mU_1|-d_F+1 = \delta_1 > 0$. The non-zero coordinates in $\Tilde{\mC}_{\mid S}$ are functions of elements from $\Tilde{\mC}_{\mid \mR_1}$, or are inherited from $\Tilde{\mC}_{\mid S_1}.$ As a result, $\dim(\Tilde{\mC}_{\mid S})$ is upper bounded by $|\mR_1 \cup S_1| \le |\mR_1|+|S_1|$. Therefore, $|\mR_1|+|S_1| \ge |\mR_1 \cup S_1| \ge k_{I_1}$. From this we obtain that $|\mR_1| \ge k_{I_1}-\delta_1$ and since $|\mU_1| \ge |\mU_1 \setminus \mR_1|$ we have $|\mR_1| \ge k_{I_1}-|\mU_1|+d_F-1$.

Finally, suppose that $d_F > n_I-k_{I_1}+1$. For the sake of contradiction, assume that $\delta_1 >0$. Arbitrarily choose $d_F-1$ elements from $T \subseteq \mU_1$ and let $S$ be the set made by removing these $d_F-1$ elements from $[n_F]$, i.e., $S=[n_F] \setminus T$. Since $|S| =n_F-d_F+1 \ge n_F-d(\Tilde{\mC})+1$, by the shortening bound, we have $\dim(\Tilde{\mC})=k_{I_1}$. The non-zero element of $\Tilde{\mC}_{\mid S}$ are either inherited from $\mU_1 \setminus T$ or functions of elements in $\mR_1$. As a result, we have $\dim(\Tilde{\mC}_{\mid S}) = k_{I_1} \le |(\mU_1 \setminus T) \cup \mR_1|$. We have
\begin{align*}
    |(\mU_1 \setminus T) \cup \mR_1| = |(\mU_1  \cup \mR_1)\setminus T| = |\mU_1  \cup \mR_1|-|T| \ge k_{I_1}.
\end{align*}
On the other hand, since $\mU_1  \cup \mR_1 \subseteq [n_F]$ we have $|\mU_1  \cup \mR_1|-|T| \le n_F-d_F+1$. This gives that $d_F  \le n_F-k_{I_1}+1$, contradicting our initial assumption.
\end{proof}
Combining Proposition~\ref{prop:unchangedupper1} with Proposition~\ref{prop:readlower} we obtain the following result.

\begin{corollary} \label{cor:readlower}
Let $(\mathcal{C}_{\mathbf{I}}, \mathcal{C}_F, \sigma)$ be a convertible code, let $i \in [\lambda]$, let $\mathcal{R}_i \subseteq [n_F]$ be the read symbols
from $\mathcal{C}_{I_i}$ and let $\omega_i:=n_F - 2d_F - \sum_{j \ne i} k_{I_j} + 2$. We have
\[
\aligned
|\mathcal{R}_i| \ge
\begin{cases}
k_{I_i}, & \text{if } \omega_i \le 0, \\[2pt]
k_{I_i} - \omega_i,
& \text{if } \omega_i > 0.
\end{cases}
\endaligned
\]
\end{corollary}

\section{Conversion of Reed-Muller Codes}\label{sec:RM}

\subsection{Definitions and Properties}
We start by giving preliminary definitions and basic properties of \emph{Reed-Muller} codes; see for example~\cite{macwilliams1977theory, abbe2020reed}. In what follows, for integers $r \ge 0$ and $m \ge 1$, we denote by
\[
\F_2[X_1,\dots,X_m]_{\le r}^\times
\]
the vector space of square-free polynomials in $X_1,\dots,X_m$ of total degree at most $r$. 
We also let $\mathcal{P}(m)$ denote the list of vectors in $\F_2^m$ sorted in lexicographic order. For example,
for $m=3$ we have
\begin{align} \label{ptsRM}
\mathcal{P}(3)=
\bigl(&(0,0,0),(0,0,1),(0,1,0),(0,1,1),
\\ &(1,0,0),(1,0,1),(1,1,0),(1,1,1)\bigr) \nonumber
\end{align}
\begin{definition}\label{def:rm}
Let $r \ge 0$ and $m \ge 1$ be integers. Let $n=2^m$ and $\mP(m)=(a_1,...,a_n)$.
The \textbf{Reed-Muller code} with parameters $(r,m)$ is
$$\RM(r,m):=\{(p(a_1),...,p(a_n)) : p \in \F_2[X_1,...,X_m]_{\le r}^\times\} \le \F_2^n.$$
\end{definition}

Some of the properties of Reed-Muller codes that we will need in the sequel are summarized as follows.

\begin{theorem} \label{thm:RMdim}
Let $r \ge 0$ and $m \ge 1$ be integers, and let $n=2^m$. 
\begin{itemize}
    \item[(i)] $\RM(r,m)$ is a linear code of dimension $\sum_{i=0}^r \binom{m}{i}$;
    \item[(ii)] The minimum distance of $\RM(r,m)$ is
$2^{m-r}$;
    \item[(iii)] For $r\ge 1$ and $m \ge 1$ we have $\RM(r,m) = \RM(r,m-1) \ps \RM(r-1,m-1),$ where for subspaces $\mC, \mD \le \F_q^n$ we let $$\mC \ps \mD:=\{(x,x+y) : x \in \mC, \, y \in \mD \} \le \F_q^{2n}$$ be the \textbf{Plotkin sum};
    \item[(iv)] The dual of $\RM(r,m)$ is $\RM(r,m)^\perp = \RM(m-r-1,m)$.
\end{itemize}
\end{theorem}


\subsection{Reed-Muller Convertible Codes}

In this subsection, we explain how to construct Reed-Muller convertible codes for $\lambda=2$. This construction can be extended recursively for larger $\lambda$.

Let $r \ge 1$ and $m \ge 1$, and denote by $G_{\RM(r,m)}$ a generator matrix of $\RM(r,m)$.

\begin{example}
Let $r=2$ and $m=3$. Then $n=8$ and $\RM(2,3)$ is generated by the evaluations 
of the polynomials in $\{1,X_1, X_2, X_3,X_1X_2,X_1X_3,X_2X_3\}$ at the points listed in~\eqref{ptsRM}. For a multivariate polynomial $p$, denote by $\textup{eval}(p)$ the vector $(p(a_1),\dots,p(a_n)) \in \F_2^n$.
Therefore a generator matrix $G_{\RM(2,3)}$ of $\RM(2,3)$ is
\begin{small}
\begin{align} \label{eq:RMex}
\begin{bmatrix}
        \textup{eval}(1) \\
        \textup{eval}(X_1) \\
        \textup{eval}(X_2) \\
        \textup{eval}(X_3) \\
        \textup{eval}(X_1X_2) \\
        \textup{eval}(X_1X_3) \\
        \textup{eval}(X_2X_3) \\
    \end{bmatrix}
    = 
    \begin{bmatrix}
1  & 1 & 1 & 1 & 1 & 1 & 1 & 1 \\
0  & 0 & 0 & 0 & 1 & 1 & 1 & 1 \\
0  & 0 & 1 & 1 & 0 & 0 & 1 & 1 \\
0  & 1 & 0 & 1 & 0 & 1 & 0 & 1 \\
0  & 0 & 0 & 0 & 0 & 0 & 1 & 1 \\
0  & 0 & 0 & 0 & 0 & 1 & 0 & 1 \\
0  & 0 & 0 & 1 & 0 & 0 & 0 & 1
\end{bmatrix} 
= \begin{bmatrix}
        G_{\RM(1,3)} \\
        \textup{eval}(X_1X_2) \\
        \textup{eval}(X_1X_3) \\
        \textup{eval}(X_2X_3) \\
    \end{bmatrix}.
\end{align}
\end{small}
\end{example}

\begin{remark} \label{rem:rmsplit}
From Theorem~\ref{thm:RMdim} (iii) we know that $\RM(r,m) = \RM(r,m-1) \ps \RM(r-1,m-1)$. In terms of generator matrices, this means that  
\begin{align*}
    G_{\RM(r,m)} = \begin{pmatrix}
        G_{\RM(r,m-1)} & G_{\RM(r,m-1)} \\
        0 & G_{\RM(r-1,m-1)}
    \end{pmatrix}.
\end{align*}

Moreover, by looking at the example in~\eqref{eq:RMex}, we can see more generally, that
\begin{align*}
    G_{\RM(r,m-1)} = \begin{pmatrix}
        G_{\RM(r-1,m-1)} \\
        A
    \end{pmatrix},
\end{align*}
where $A$ is a matrix formed by evaluating the monomials $p \in \F_2[X_2,...,X_m]_{ \le r}^\times$ of degree~$r$ at the points in $\mP(m-1)$. In particular, $A$ has $\sum_{i=0}^{r-1} \binom{m-1}{i}$ zero columns (because any vector in $\mP(m-1)$ of weight at most $r-1$ will be 0 when evaluated at a monomial of degree $r$). Therefore, by performing row operations on $G_{\RM(r,m)}$ we obtain a matrix of the form
\begin{align} \label{eq:rmm}
    \Tilde{G}_{\RM(r,m)}= \begin{pmatrix}
        G_{\RM(r,m-1)} & 0 \\
         & A \\
        0 & G_{\RM(r-1,m-1)}
    \end{pmatrix}
\end{align}
which clearly also is a generator matrix of $\RM(r,m)$.
\end{remark}


\begin{theorem} \label{thm:rmc}
    Let $r \ge 1$ and $m \ge 1$. Let $\mC_{I_1}=\RM(r,m-1)$, $\mC_{I_2}=\RM(r-1,m-1)$, $\mC_{\textbf{I}}=\mC_{I_1} \times \mC_{I_2}$, and $\mC_F=\RM(r,m)$. There exists a convertible code $(\mC_{\textbf{I}},\mC_F,\sigma)$ with 
    \begin{align*}
        &|\mU_1| = n_{I_1}, \quad |\mU_2| = k_{I_2}, \\
        &|\mR_1| \le  k_{I_1} \quad |\mR_2| = \min\{k_{I_2},n_{I_2}-k_{I_2}\}.
    \end{align*}
    
\end{theorem}
\begin{proof}
    We use the form of the generator matrix of $\RM(r,m)$ as in~\eqref{eq:rmm} adopting also the notation used in Remark~\ref{rem:rmsplit}. Note that we can leave all symbols from $\mC_{I_1}$ unchanged, i.e., $|\mU_1|  = n_{I_1}$. Now $A$ has $\sum_{i=0}^{r-1} \binom{m-1}{i}=k_{I_2}$ zero columns. Thus we can leave those $k_{I_2}$  zero columns from $\mC_{I_2}$ unchanged, and we have to read $\min\{k_{I_2},n_{I_2}-k_{I_2}\}$ symbols in order to generate the leftover $n_{I_2}-k_{I_2}$ columns. In order to generate the  $n_{I_2}-k_{I_2}$ non-zero columns of $A$, we need to read at most $k_{I_1}$ symbols from $\mC_{I_1}$. 
\end{proof}

\section{Comparison}\label{sec:comp}

In this section, we consider the convertible code from Theorem~\ref{thm:rmc} for $m = r + 2 \ge 4$, and compare its read and write costs to the bounds in Subsections~\ref{subsec:unchanged} and~\ref{subsec:read}. To this end, we state the corresponding bounds from Propositions~\ref{prop:unchangedupper1},~\ref{prop:unchangedupper2}, and~\ref{prop:readlower}, and derive the conditions under which these bounds are met with equality.
    \begin{align*}
        &n_{I_1}=n_{I_2}=2^{m-1}, \quad k_{I_1} = \sum_{i=0}^{r} \binom{m-1}{i}=2^{m-1}-1,\\
        &k_{I_2} = \sum_{i=0}^{r-1} \binom{m-1}{i} = 2^{m-1}-m, \quad n_{F}=2^{m},\\
        &k_{F} = \sum_{i=0}^{r} \binom{m}{i},\quad d_F=2^{m-r}=2^2, \quad d_F^\perp = 2^{r+1}=2^{m-1}.
    \end{align*}
    The bound of Proposition~\ref{prop:unchangedupper1} evaluated at the above parameters and for $m=r+2$ reads as follows:
\[
\scalebox{0.85}{$
\aligned
|\mathcal{U}_1|
&\le 2^m - 2^2 - \sum_{i=0}^{m-3} \binom{m-1}{i} + 1
= 2^{m-1} + m - 3
\Longrightarrow |\mathcal{U}_1| \le n_{I_1}, \\
|\mathcal{U}_2|
&\le 2^m - 2^2 - \sum_{i=0}^{m-2} \binom{m-1}{i} + 1
= 2^{m-1} - 2^2
\Longrightarrow |\mathcal{U}_2| \le n_{I_2} - 4 .
\endaligned
$}
\]
Now note that $$d_F^\perp = 2^{m-1} \le (2^{m-1}-1)+1 = k_{I_{1}}+1,$$ so the bound of Proposition~\ref{prop:unchangedupper2} is not applicable to $\mC_{I_1}$. For $\mC_{I_2}$ we have $$d_F^\perp = 2^{m-1} > (2^{m-1}-m)+1 = k_{I_{2}}+1,$$ so the bound of Proposition~\ref{prop:unchangedupper2} applies, and we have $|\mU_2| \le k_{I_2}$. This shows that the Reed-Muller convertible code meets the bounds on the number of unchanged symbols with equality, and thus has optimal write cost.

For the read cost, we evaluate the bound of Proposition~\ref{prop:readlower}. We have $\delta_1 = n_{I_1}-d_F+1=2^{m-1}-2^2+1 \ge 0$ and so $\mR_1 \ge k_{I_1}-2^{m-1}+2^2-1=2$. This means for $\mR_1$, the bound is not met with equality by our construction. 

We have $\delta_2 = k_{I_2}-d_F+1=2^{m-1}-m-2^2+1 > 0$ for $m \ge 4$ and so $\mR_2 \ge 3$. Since $\min\{n_{I_2}-k_{I_2},k_{I_2}\}=\min\{m,2^{m-1}-m\}=m$. Therefore, this bound is not met with equality for our construction.

\begin{remark}
    The above example shows that for the case where $m = r+2 \ge 4$, our construction of Reed--Muller convertible codes meets the bounds from Subsection~\ref{subsec:unchanged} for $|\mU_1|$ and $|\mU_2|$ with equality. In contrast, for $|\mR_1|$ and $|\mR_2|$, the bounds from Subsection~\ref{subsec:read} are not met with equality by our construction. Since our bounds on the number of unchanged symbols incorporate substantial structural properties of our intial and final codes, whereas the bounds on the number of read symbols are more coarse and lose exactness through some rough estimates in their proof, it would be interesting to derive bounds on~$|\mR_1|$ and~$|\mR_2|$ that take into account more structural information about our codes.

\end{remark}

\section{Discussion and Future Directions}\label{sec:concl}
In this paper, we connect the notions of server heterogeneity and data heterogeneity by proposing convertible Reed–Muller codes in the merge regime for the first time, with a focus on optimizing access costs. To evaluate the performance of the conversion process, we derive bounds for general linear codes and show that, in specific cases, our constructions partially meet these bounds. Future research directions include tailoring these bounds to Reed–Muller codes by leveraging their locality and availability properties to achieve tightness across all parameters, as well as considering the optimization of bandwidth and non-linear conversion.


\newpage

\balance

\printbibliography

\end{document}